\documentclass[conference,letterpaper]{IEEEtran}
\usepackage{threeparttable}
 \usepackage{physics}
 \usepackage{varwidth} 
 \usepackage{amsthm}

\makeatletter
\newtheorem*{rep@theorem}{\rep@title}
\newcommand{\newreptheorem}[2]{%
\newenvironment{rep#1}[1]{%
 \def\rep@title{#2 \ref{##1}}%
 \begin{rep@theorem}}%
 {\end{rep@theorem}}}
\makeatother

\makeatletter
\newtheorem*{rep@claim}{\rep@title}
\newcommand{\newrepclaim}[2]{%
\newenvironment{rep#1}[1]{%
 \def\rep@title{#2 \ref{##1}}%
 \begin{rep@claim}}%
 {\end{rep@claim}}}
\makeatother

\newreptheorem{theorem}{Theorem}
\newreptheorem{claim}{Claim}
\newreptheorem{lemma}{Lemma}

\usepackage[utf8]{inputenc}
\usepackage{amsmath,bm}
\usepackage{amsthm}
\usepackage{amssymb}
\usepackage{mathtools}
\usepackage{bbm}
\allowdisplaybreaks
\usepackage{graphicx}
\usepackage[dvipsnames]{xcolor}
\usepackage[english]{babel}
\usepackage[font=small]{caption}
\usepackage{subcaption}
\usepackage{wrapfig}
\usepackage{tabularx}
\usepackage{multirow}
\usepackage{multicol}
\usepackage{thmtools, thm-restate}
\usepackage[normalem]{ulem}
\usepackage[noadjust]{cite}
\usepackage{makecell}
\usepackage{hyphenat}
\usepackage{tikz-cd}
\usepackage{enumitem}

\usepackage[final]{changes}

\newtheorem{theorem}{Theorem$\!$}
\newtheorem{example}[theorem]{Example$\!$}
\newtheorem{lemma}[theorem]{Lemma$\!$}

\newtheorem{definition}[theorem]{Definition$\!$}

\usepackage{graphicx}

\usepackage{xcolor}

\newcommand{\sseneq}{\xrightarrow{\vbox{\hbox to 0.32cm{\,\customneq}\kern-0.5ex}}}

\newcommand{\cC}{\mathcal{C}}

\newcommand{\cP}{\mathcal{P}}

\newcommand{\mybold}[1]{\bm{#1}}

\newcommand{\bd}{{\mybold{d}}}

\newcommand{\bsf}{{\mybold{f}}}

\newcommand{\bl}{{\mybold{l}}}

\newcommand{\bp}{{\mybold{p}}}

\newcommand{\br}{{\mybold{r}}}
\newcommand{\bs}{{\mybold{s}}}

\newcommand{\bu}{{\mybold{u}}}
\newcommand{\bv}{{\mybold{v}}}
\newcommand{\bw}{{\mybold{w}}}
\newcommand{\bx}{{\mybold{x}}}
\newcommand{\by}{{\mybold{y}}}
\newcommand{\bz}{{\mybold{z}}}

\newcommand{\one}{\mathbf{1}}
\newcommand{\zero}{\mathbf{0}}

\newcommand{\balpha}	{\mybold{\alpha}}
\newcommand{\bbeta}		{\mybold{\beta}}

\DeclareMathOperator{\VT}{VT}
\DeclareMathOperator{\DVT}{\Delta\!VT}

\usepackage{breqn}
\ifCLASSINFOpdf
\else
\fi
\usepackage[utf8]{inputenc} 
\usepackage[T1]{fontenc}
\usepackage{url}
\usepackage{ifthen}
\usepackage{cite}
\usepackage{amsmath} 

\newcommand{\customneq}{
  \tikz[scale=0.5,baseline=-0.5ex,line width=0.6pt]{
    \draw (0,0) -- (0.36,0);
    \draw (0,-0.12) -- (0.36,-0.12); 
    \draw (0.08,-0.22) -- (0.28,0.09); 
  }
}

\usepackage{cleveref}
\begin{document}
\title{Asymptotically Optimal Codes Correcting One Substring Edit} 



\author{%
  \IEEEauthorblockN{Yuting Li\IEEEauthorrefmark{1},
                    Yuanyuan Tang\IEEEauthorrefmark{2},
                    Hao Lou\IEEEauthorrefmark{2},
                    Ryan Gabrys\IEEEauthorrefmark{3},
                    and Farzad Farnoud\IEEEauthorrefmark{1}\IEEEauthorrefmark{2}}
  \IEEEauthorblockA{\IEEEauthorrefmark{1}%
                    Computer Science, University of Virginia, USA, \texttt{mzy8rp@virginia.edu}}
  \IEEEauthorblockA{\IEEEauthorrefmark{2}%
                    Electrical \& Computer Engineering,
                    University of Virginia, USA,   \\ \texttt{\{yt5tz,hl2nu,farzad\}@virginia.edu}}
  \IEEEauthorblockA{\IEEEauthorrefmark{3}%
                    Calit2, University of California-San Diego, USA, \texttt{rgabrys@ucsd.edu}}
}

\maketitle


\begin{abstract}
    The substring edit error is the operation of replacing a substring $\bu$ of $\bx$ with another string $\bv$, where the lengths of $\bu$ and $\bv$ are bounded by a given constant $k$. It encompasses localized insertions, deletions, and substitutions within a window. Codes correcting one substring edit have redundancy at least $\log n+k$.
In this paper, we construct codes correcting one substring edit with redundancy $\log n+O(\log \log n)$, which is asymptotically optimal. The full version of this paper is available online.\footnote{\url{https://www.ece.virginia.edu/~ffh8x/papers/eccss.pdf}}
\end{abstract}

\section{Introduction}\label{sec:intro}

In data transmission and storage, especially in data storage in DNA~\cite{tabatabaeiyazdi2015} and other emerging media, the need for robust error correction for a diverse set of errors is an important and challenging problem. This paper presents a family of asymptotically optimal codes capable of correcting a burst of localized edits, encompassing combinations of insertions, deletions, and substitutions. In addition to their application in error-correction, the codes provide a promising approach to the challenge of synchronizing multiple copies of related files, where bursts of edits are common. The data synchronization problem, also known as document exchange, has been explored in various works demonstrating that error-correcting codes, especially those accommodating file edits, can substantially reduce the communication required to maintain file consistency~\cite{orlitsky1991b}.

Specifically, we focus on correcting a single burst of errors of length at most $k$, which we refer to as a \mbox{$k$-\textit{\textbf{substring edit}}}. Such an error may be the result of insertions, deletions, or substitutions, occurring within a bounded interval of our strings. Formally, a $k$-substring edit in a string $\bx$ is the operation of replacing a substring $\bu$ of $\bx$ with another string $\bv$, where $|\bu|,|\bv|\le k$. An analysis of real and simulated data in~\cite{tang2023} reveals that substring edits are common in file synchronization and DNA storage, where viewing errors as substring edits leads to smaller redundancies. 
%

In this paper, we assume that $k$ is a fixed constant. Notice that for the setting where we replace a string $\bu$ with another string $\bv$ of the same length (i.e., $\abs{\bu} =\abs{\bv}$), a substring edit is equivalent to a burst of substitutions. Furthermore, if $\bv$ is the empty string, then a substring edit is equivalent to a burst of deletions. Analogously, if $\bu$ is the empty string, then the substring edit is a burst of insertions. Despite the fact that these specialized types of substring edit have each received significant attention in the past, the problem of constructing codes for substring edits has received relatively less attention, with the exception of \cite{tang2023} and~\cite{bours1994codes}. A critical observation in this context is that the ability to correct an arbitrarily long burst of deletions does not enable correcting a $k$-substring edit~\cite[Lemma~1]{tang2023}.   

Using simple counting arguments, it can be shown that a code correcting one substring edit has redundancy at least $\log n+k$ and at most roughly $2 \log n$.
In \cite{tang2023}, the authors present an efficient construction for a code that corrects one substring edit with redundancy roughly $2\log n$, which matches the existential upper bound of redundancy. However, the question of whether it is possible to construct a code with less than $2 \log n$ bits of redundancy remained open. In this paper, we provide an affirmative answer to this question and construct codes correcting one substring edit with redundancy roughly $\log n+O(\log\log n)$, i.e., at most $O(\log\log n)$ bits larger than the optimal redundancy. Hence, the codes we propose are asymptotically  optimal in terms of redundancy and improve the state-of-the-art redundancy by a factor of 2. 

The rest of the paper is organized as follows: In Section~\ref{sec:prelim}, we introduce basic definitions and notation, review our techniques, and compare with related works. In Section~\ref{sec:2}, we construct asymptotically optimal codes correcting one substring edit.


\section{Notation, Overview, and Related Work}\label{sec:prelim}
\subsection{Basic Definitions and Notations}
We will deal with two types of strings, strings over $\{0,1\}$ and strings over the set of nonnegative integers. When necessary to identify the alphabet, we refer to a string as a \emph{binary string} or an \emph{integer string}, respectively. Strings are denoted by bold symbols, e.g., $\bu=u_1\dotsm u_n$. Let $[i,j]$ represent the integers $i,i+1, \dotsc, j$. Furthermore, the set $[1,j]$ is denoted as $[j]$.
The substring of $\bu$ starting at position $i$ and ending at position $j$ is denoted by $\bu_{[i,j]}$. 
We use $\abs{\bu}$ to denote the length of $\bu$. For $\bz\in \mathbb{N}^*$ and a positive integer $A$, we say that $\bz$ is $A$-bounded if each symbol of $\bz$ is at most $A$. 

For binary strings $\bx,\by\in\{0,1\}^*$, we denote a $k$-substring edit by $\bx\xrightarrow[]{}\by$. If $\bx\xrightarrow[]{}\by$ and $\bx\neq \by$, then we may also write $\bx\sseneq \by$. In our construction, substring edits over integer strings also arise. For clarity, we use a different notation for these: For two integer strings $\bz,\bw\in\mathbb N^*$, we let $\bz\Rightarrow\bw$ denote a substring edit (not necessarily a $k$-substring edit) and let $\bz\xRightarrow{j:a,b}\bw$ denote the substring edit that deletes a substring of length $a$ starting in position $j$, and inserts a substring of length $b$ in its place. This operation is a $\max(a,b)$-substring edit. 

A partition $\cP$ of $\bx$ is a rule under which one divides $\bx$ into parts. We use $n_{\cP}(\bx)$ to denote the number of parts and $\bx^{\cP}$ to denote the vector $(\bx_1,\bx_2,\dotsc,\bx_{n_{\cP}(\bx)})$, where $\bx_i$, $1\leq i\leq n_{\cP}(\bx)$, denotes a part. We call a string $(\cP,\delta)$-dense if all elements of $\bx^{\cP}$ have length at most $\delta$.
Suppose $f$ is a function on strings. Then we use $\bsf^{\cP}(\bx)$ to denote the string $f(\bx_1)f(\bx_2)\dotsm f(\bx_{n_{\cP}(\bx)})$. When the partition $\cP$ is clear from the context, we may simply use $\bsf(\bx)$ to denote $\bsf^{\cP}(\bx)$. (Note that $f(\bx)$, in contrast to $\bsf(\bx)$, is simply the result of applying $f$ to $\bx$.)
For a (binary or integer) string $\bz$, its VT sketch is defined as
$\VT(\bz)=\sum_{i=1}^{\abs{\bz}} i z_i $.

\subsection{Overview of the Techniques}\label{ssc:overview}
The overview of our approach to correcting one substring edit error is given in Figure~\ref{fig:figure1_}. Here, $\bx$ is the input to the channel while $\by$ is the output. The first step in correcting the error is to approximately identify its position. To achieve this, we rely on a key insight that goes back to the pioneering work of Levenshtein~\cite{levenshtein1967asymptotically} on correcting a burst of at most 2 deletions. Namely, while the VT sketch applied to $\bx$ can locate only a single deletion, it can be more powerful when applied to carefully designed functions of $\bx$. 
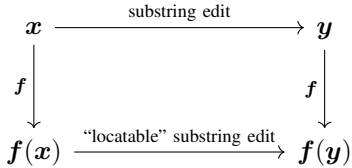
\begin{figure}
\centering
\begin{tikzcd}[column sep=80pt,row sep=30pt]
\bx \arrow[r,"\text{substring edit}"] \arrow[d,"\bsf"'] &
\by \arrow[d,"\bsf"'] 
\\
\bsf(\bx) \arrow[r,"\text{``locatable'' substring edit}"] &
\bsf(\by) 
\end{tikzcd}
\caption{converting a substring edit to a ``locatable'' substring edit}
\label{fig:figure1_}
\end{figure}
\label{subsubsec:tech1}
Previous work has used this observation to correct deletions within a window, as well as to correct one deletion or one adjacent transposition\cite{lenz2020optimal}\cite{bitar2021optimal}\cite{Gabrysbeyond2023}. 
Our main contribution in this direction is to introduce a novel mapping that, when used in conjunction with a VT sketch, can locate the position of any $k$-substring edit to within a small interval. In particular, this mapping allows us to locate a burst of substitutions, 
a challenging task since such an edit may leave the sequence composed of lengths between ``markers'' untouched. 

To construct the aforementioned mapping, denoted in Figure~\ref{fig:figure1_} by $\bsf$, we identify a set of conditions on a substring edit over an integer string that, if satisfied, the location of the substring edit can be approximately identified using the VT sketch and some other small amount of information (Lemma~\ref{lem:locationA}). Such a substring edit is called ``locatable'', a term that is defined precisely in Definition~\ref{def:locatable}. 
Specifically, our approach will be to find a partition rule $\cP$ and a function $f$ such that if $\bx\sseneq \by$, then $\bsf(\by)=\bsf^{\cP}(\by)$ is obtained from $\bsf(\bx) = \bsf^{\cP}(\bx)$ by a locatable substring edit. 
Then we restrict $\bx$ to be $(\cP,\delta)$-dense. 
If $\bx\sseneq \by$, we treat $\bsf^{\cP}(\bx)$ and $\bsf^{\cP}(\by)$ as the $\bz$ and $\bw$ in Lemma~\ref{lem:locationA}, and determine the edit position of this locatable substring edit to within an interval of length $O(\delta)=O(\log n)$. Since we restrict $\bx$ to be $(\cP,\delta)$-dense, which means $\abs{\bx_i}\leq \delta$ for any $i$, we can determine the original edit position to within an interval of length $O(\delta^2)$ (a technique used in \cite{bitar2021optimal}), which is then corrected. The case in which $\bx=\by$ will also be handled in \Cref{thm:1substringedit}. The VT sketch used in our construction will require a redundancy of about  $\log n$ bits, while the other components will have negligible redundancy.

\subsection{Related Works}
Codes correcting localized errors have been a subject of extensive study over time. Fire codes \cite{blahut2003algebraic} are designed to construct a burst of $k$ substitutions with roughly $\log n$ redundancy. Several works, including \cite{bours1994codes,Cheng2014,schoeny2017codes}, have construct burst-deletion correcting codes, where the length of the bursts is known in advance. More related to this work are codes that can correct a variable number of deletions. In a pioneering work by Levenshtein's work \cite{levenshtein1967asymptotically}, optimal codes are constructed to correct a burst of at most two consecutive deletions. 

For longer deletions and non-consecutive but localazide deletions, several works \cite{lenz2020optimal,bitar2021optimal,Gabrysbeyond2023,tang2023} aim to first approximately locate the deletion. Their method uses the VT sketch of an integer string that is related to the original binary string to locate the approximate deletion position. Specifically, in \cite{lenz2020optimal}, \cite{bitar2021optimal}, and \cite{tang2023}, each symbol of the integer string is the length of each part of the original binary string with respect to some partition.
Using this method, Lenz et al. \cite{lenz2020optimal} propose codes correcting a variable-length burst of deletions with nearly optimal redundancy, but the deletions still need to be consecutive. Bitar et al. \cite{bitar2021optimal} later construct nearly optimal codes that remove the restriction of consecutive deletions, so that it could correct deletions within a window that do not necessarily occur consecutively. Also using VT sketch of the length integer string, Tang et al. \cite{tang2023} construct codes correcting a substring edit, but the VT sketch of the length integer string can locate the edit position only when the substring edit changes the length of the original binary string. So the construction from \cite{tang2023} needs another $\log n$ bits to handle the other cases, and the total redundancy is about $2\log n$ bits, which is not optimal.
 By using a different integer string, Gabrys et al. \cite{Gabrysbeyond2023} further construct codes correcting one deletion
or one adjacent transposition with near-optimal redundancy, but their approach does not address the setup where additional edits may occur.


\section{Asymptotically Optimal Codes Correcting One Substring Edit}\label{sec:2}
In this section, we present the code construction, prove its correctness, and determine its redundancy. Based on our earlier discussion in Subsection~\ref{ssc:overview}, the construction has three interrelated components: $a$) identifying a set of conditions that make a substring edit over an integer string ``locatable'' and showing that the position of a locatable substring edit can be found approximately; $b$) constructing a partition rule $\cP$ and a mapping $f$ such that if $\bx\sseneq\by$, then $\bsf^{\cP}(\by)$ is obtained from $\bsf^{\cP}(\bx)$ through a ``locatable'' substring edit; and $c$) correcting the edit in $\by$ and recovering $\bx$. We discuss each in a subsection below. In the next subsection, we discuss each of these components. Then we prove the correctness of the overall approach.

\subsection{Components of the Code Construction}\label{subsec:components}



\paragraph{Locatable substring edits}
We will now define the notion of locatable substring edit in Definition~\ref{def:locatable}, and prove Lemma~\ref{lem:locationA}, which roughly says that for two $\delta$-bounded integer strings $\bz$ and $\bw$, if $\bw$ is obtained from $\bz$ by a locatable substring edit at position $j$, then given $\bw$, $\VT(\bz)$, and some other little information, one can determine $j$ to within an interval of length $O(\delta)$. 

\begin{definition}\label{def:locatable}
    For two positive integer strings $\bz$ and $\bw$, we say that $\bw$ is obtained from $\bz$ by a \emph{\textbf{locatable}} $K$-substring edit at position $j$ if $\bz\xRightarrow{j:a,b}\bw$, where $a,b\leq K$, and
    \begin{enumerate}
    \item The difference of the sum of $\bz$ and the sum of $\bw$ is bounded by all $w_i$, namely $\abs{\sum_{j=1}^{\abs{\bz}}z_j-\sum_{j=1}^{\abs{\bw}}w_j}< w_i$ for all $i$'s.
    \item $\abs{\bz}\neq \abs{\bw}$ or $\sum_{j=1}^{\abs{\bz}}z_j\neq\sum_{j=1}^{\abs{\bw}}w_j$.
    \end{enumerate}
\end{definition}

Now we show that for $A$-bounded strings $\bz$ and $\bw$, if $\bw$ is obtained from $\bz$ by a locatable $K$-substring edit at position $j$, then one can determine $j$ to within an interval of length $O(A)$ given $\VT(\bz)\bmod O(m)$ (for some parameter $m$ specified in the lemma statement) and some other little information of $\bz$.

\begin{lemma}\label{lem:locationA}
For $A$-bounded strings $\bz$ and $\bw$, if $\bw$ is obtained from $\bz$ by a locatable $K$-substring edit at position $j$, then
 one can determine $j$ to within an interval of length $6K^2A$ given $\bw$, $\abs{\bz}-\abs{\bw}$, $\sum_{i=1}^{\abs{\bz}}z_i-\sum_{i=1}^{\abs{\bw}}w_i$, and $\VT(\bz)\bmod m$, where $m>8K^2\left(\sum_{i=1}^{\abs{\bz}}z_i+A\right)$.
\end{lemma}
\begin{proof}
Suppose $\bz\xRightarrow[]{j:a,b}\bw$, where $a,b\leq K$ and $j\in [1,\abs{\bz}-a+1]$.
Let $\eta(v,\bz,\bw):= (\abs{\bz}-\abs{\bw})\sum_{i= v}^{\abs{\bw}}w_i+(\sum_{i=1}^{\abs{\bz}}z_i-\sum_{i=1}^{\abs{\bw}}w_i) v,$ where $v$ is an integer. Note that since we know $\bw$, $\abs{\bz}-\abs{\bw}$, and $\sum_{i=1}^{\abs{\bz}}z_i-\sum_{i=1}^{\abs{\bw}}w_i$, we know the expression of $\eta(\cdot,\bz,\bw)$. Let $\DVT:= \VT(\bz)-\VT(\bw)$. We will prove the following. 

\begin{enumerate}
    \item $\eta(\cdot,\bz,\bw)$ is strictly monotone on $[1,\abs{\bw}+1]$.
    \item $\abs{\DVT-\eta(j,\bz,\bw)}< 3K^2A$.
    \item    $\abs{\DVT}\leq 4K^2\left(\sum_{i=1}^{\abs{\bz}}z_i+A\right).$
\end{enumerate}

For 1), note that 
\begin{dmath*}
    \eta(v,\bz,\bw)-\eta(v-1,\bz,\bw)=(\abs{\bz}-\abs{\bw})w_v+\sum_{i=1}^{\abs{\bz}}z_i-\sum_{i=1}^{\abs{\bw}}w_i.
\end{dmath*}
If $\abs{\bz}\neq \abs{\bw}$, since $\bw$ is obtained from $\bz$ by a locatable $K$-substring edit, we have $\abs{\sum_{i=1}^{\abs{\bz}}z_i-\sum_{i=1}^{\abs{\bw}}w_i}<w_v$. 
So $\eta(\cdot,\bz,\bw)$ is strictly monotone on $[1,\abs{\bw}+1]$.
If $\abs{\bz}=\abs{\bw}$, since $\bw$ is obtained from $\bz$ by a locatable $K$-substring edit, we have $\abs{\sum_{i=1}^{\abs{\bz}}z_i-\sum_{i=1}^{\abs{\bw}}w_i}\neq 0$. So, $\eta(\cdot,\bz,\bw)$ is strictly monotone on $[1,\abs{\bw}+1]$.

For 2), since $\bz\xRightarrow[]{j:a,b}\bw$, we have 

        \begin{dmath*}
            \DVT
        =f(j,\bz,\bw)+\sum_{i=1}^{a-1}iz_{j+i}-\sum_{i=1}^{b-1}iw_{j+i}-(\abs{\bz}-\abs{\bw})\sum_{i= j}^{j+b-1}w_i.
        \end{dmath*}
    Since $\bz$ and $\bw$ are $A$-bounded strings, 
    we have $\abs{\sum_{i=1}^{a-1}iz_{j+i}}\leq a^2A$, $\abs{\sum_{i=1}^{b-1}iw_{j+i}}\leq b^2A$, and $\abs{\sum_{i= j}^{j+b-1}w_i}\leq bA$. So we have 
    \begin{dmath}\label{eq:1}
        \abs{\eta(j,\bz,\bw)-\DVT}\leq (a^2+b^2)A+\abs{a-b}bA< 3K^2A.        
    \end{dmath}

For 3), since $\eta(\cdot,\bz,\bw)$ is monotone on $[1,\abs{\bw}+1]$, 
\begin{align}\label{eq:aaa}    \abs{\eta(j,\bz,\bw)}\leq&\abs{\eta(1,\bz,\bw})+\abs{\eta(\abs{\bw}+1,\bz,\bw)}\\
    \leq & K^2\left(\sum_{i=1}^{\abs{\bz}}z_i+A\right).\nonumber
\end{align}

    Thus, by \eqref{eq:1} and \eqref{eq:aaa}, we have
    \begin{dmath*}
        \abs{\DVT}\leq \abs{\DVT-\eta(j,\bz,\bw)}+\abs{\eta(j,\bz,\bw)}\leq 4K^2\left(\sum_{i=1}^{\abs{\bz}}z_i+A\right).
    \end{dmath*}
Now since we know $\VT(\bz)\bmod m$, we can find $\DVT$. Since $\eta(\cdot,\bz,\bw)$ is strictly monotone on $[1,\abs{\bw}+1]$ and 
     $\abs{\DVT-\eta(j,\bz,\bw)}< 3K^2A$, we can determine $j$ to within an interval of length $6K^2A$.
\end{proof}

\paragraph{Constructing a mapping to convert a general substring edit to a locatable edit}
Our goal here is to construct a partition $\cP$ and a function $f$ over strings such that if $\bx\sseneq\by$, then $\bsf(\by)=\bsf^{\cP}(\by)$ is obtained from $\bsf(\bx)=\bsf^{\cP}(\bx)$ by a locatable substring edit. We first give the partition rule $\cP$. Let $\bp=0^k1^k$, which we call a pattern. 

We define the partition rule $\cP$ as follows.
For $\bx\in \{0,1\}^*$, let $\bx^{\cP}=(\bx_1,\bx_2,\dotsc,\bx_{n_{\cP}(\bx)})$, where the starting position of $\bx_1$ is the leftmost bit of $\bx$ and the starting position of $\bx_i$ ($i\geq 2$) is the leftmost bit of the $(i-1)$th $\bp$ in $\bx$. 

\begin{example}
    Let $k=3$, $\bp=0^31^3$, and$$\bx=\overbrace{01001}^{\bx_1}\overbrace{\underbrace{000111}_{\bp}100}^{\bx_2}\overbrace{\underbrace{000111}_{\bp}101001}^{\bx_3}\overbrace{\underbrace{000111}_{\bp}1}^{\bx_4}.$$ To get $\bx^{\cP}$, we first find all the patterns, and break it up at the left side of each pattern. So we get $\bx^{\cP}=(01001,000111100,000111101001,0001111)$.
\end{example}
We now explain the main idea behind the construction of the function $f$. 
For any lowercase function $u$ over strings, we use the uppercase function $U$ to denote the sum of all the symbols of $\bu(\bx)$, namely 
$$U(\bx):= \sum_{i=1}^{n_{\cP}(\bx)} u(\bx_i) = \sum_{i=1}^{|\bu(\bx)|} u(\bx_i),$$
and use $\Delta U$ to denote $U(\bx)-U(\by)$. For simplicity, we only consider the second condition of the locatable substring edit (see the full version for a complete treatment). Applied to the substring edit $f(\bx)\Rightarrow f(\by)$, the second condition of \Cref{def:locatable} requires that if $\bx\sseneq \by$, then $\abs{\bsf(\bx)}\neq \abs{\bsf(\by)} \text{ or } \Delta F\neq 0$
. Note that if the number of patterns changes ($n_{\cP}(\bx)\neq n_{\cP}(\by)$), then $\abs{\bsf(\bx)}\neq \abs{\bsf(\by)}$. Hence, we assume that the number of patterns does not change. We want to ensure that $\Delta F=F(\bx)-F(\by)\neq 0$. To better explain our approach, we first outline two different approaches that have increasing levels of complexity before arriving at our actual solution, which is described as the third attempt.
 
\textbf{First attempt.}
For a string $\bx$, let $l(\bx)$ represent the length of $\bx$. Note that $L(\bx)$ is also the length of $\bx$. In our first attempt, we let $f(\cdot)=l(\cdot)$. Then $F(\bx)=L(\bx)$. In this case, if the lengths of the deleted and inserted strings differ, then $\Delta F\neq 0$. However, $\Delta F=0$ occurs when the substring edit is a burst of substitutions. In fact, in \cite{tang2023}, the authors use $f(\cdot)=l(\cdot)$, and use another $\log n$ bits of redundancy to handle a burst of substitutions separately. This leads to a redundancy of $2\log n$ bits. In this work, we will introduce two other novel functions, $d$ (second attempt), and $n_{1^k}$ (third attempt), to resolve this issue.

\textbf{Second attempt.}
Let $d:\{0,1\}^n\rightarrow [2^k]$ be a function that can detect a burst of $k$ substitutions, where the $j$th bit of $d(\bx)$ is the parity of every $k$th index of $\bx$ stating at position $j$, namely, 
$$d(\bx)
:=\left(\bigoplus_{i\equiv_k 1 }x_i,\bigoplus_{i\equiv_k 2 }x_i,\dotsc,\bigoplus_{i\equiv_k k}x_i\right),$$ 
where $\oplus$ is addition in $\mathbb F_2$. Note that $\abs{\Delta D}\leq 2^k$.
We let $f(\cdot)=d(\cdot)+Cl(\cdot)$, where $C$ is a sufficiently large constant. Then, $\Delta F=\Delta D+C\Delta L$. 

Now, if $\Delta L\neq 0$, then $\Delta F\neq 0$ provided that we choose $C> 2^k$. If $\Delta L=0$, then the edit transforming $\bx$ into $\by$ is a burst of substitutions. At first glance, it may seem that in this case $\Delta D\neq 0$, and thus $\Delta F \neq 0$. But while this is the case if this burst of substitutions does not involve any patterns, it may not be the case otherwise. The next example shows that  if this burst of substitutions involves some patterns, then $\Delta D$ may be $0$. 
\begin{example}\label{ex:notwork}
     Let $k=3$, $\bp=0^31^3$, and 
     \begin{align*}     \bx=&\overbrace{000011}^{\bx_1}\overbrace{\underbrace{000\textcolor{red}{\one\one\one}}_{\bp}001111}^{\bx_2},\\   \by=&\overbrace{000011000\textcolor{red}{\zero\one}}^{\by_1}\overbrace{\underbrace{\textcolor{red}{\zero}00111}_{\bp}1}^{\by_2},\end{align*}
     where the edited substring is in bold and red.
    Then $\bl(\bx)=l(\bx_1)l(\bx_2)=(6,12)$, $\bl(\by)=(11,7)$.
     $\bd(\bx)=(011,001)=(3,1)$, $\bd(\by)=(001,011)=(1,3)$.
Then $\Delta L=0$ and $\Delta D=0$, so $\Delta F=0$. Thus, $\bsf(\bx)\Rightarrow\bsf(\by)$ is not a locatable substring edit.
\end{example}
\textbf{Third attempt.}
Recall that the method in the second attempt may not work only if the substring edit is a burst of substitutions that involves some patterns. By careful argument, one can find that the only case where $\Delta F=0$ in the second attempt is when one pattern $0^k1^k$ is shifted because of this burst of substitutions. So, we need a function, say $g$, such that $\Delta G\neq 0$ in this case, and thus our $f$ will be of the form $f(\cdot)=d(\cdot) + B g(\cdot) + C l(\cdot)$, where $B$ is a sufficiently large constant and $C$ is much larger than $B$. This way, if $\Delta L\neq 0$, then $\Delta F\neq 0$. If $\Delta L=0$, and $\Delta G\neq 0$, then $\Delta F\neq 0$. If $\Delta L=0$ and $\Delta G =0$, then $\bx\sseneq \by$ is a burst of substitutions that does not involve any pattern, then $\Delta D\neq 0$, thus $\Delta F\neq 0$. 

Let $n_{1^k}(\bx)$ be the number of appearances of $1^k$ in $\bx$. For example, if $\bx=111100111$, then $n_{1^2}(\bx)=5$. We will show that $n_{1^k}$ is our desired $g$. The next example provides the intuition. Note that  $N_{1^k}(\bx)$ is the sum of $n_{1^k}(\bx_i)$ over $\bx^{\cP}=(\bx_1,\dotsc,\bx_{n_{\cP}(\bx)})$.



\begin{example}\label{ex:shift}

Let $N_{1^k}^{L}(\bx)$ and $N_{1^k}^{R}(\bx)$ denote the contribution of the left side and the right side with respect to the `$|$' symbols below. Similar definitions apply to  $\by$. Below, we show two typical cases in which $\bx\sseneq \by$ is a burst of substitutions that shifts a pattern. In both cases, we show that $N_{1^k}(\bx)\neq N_{1^k}(\by)$, implying that $\Delta N_{1^k}\neq 0$. 

First, suppose $\bx=\balpha0^{\epsilon}|0^k1^k\bbeta$ and $\by=\balpha|0^{k}1^{k}1^{\epsilon}\bbeta$, for some binary string $\balpha$ and $\bbeta$. Here $0<\epsilon\leq k$. The burst of substitutions occurs at positions $[\abs{\balpha}+k+1,\abs{\balpha}+k+\epsilon]$. Below is an example with $k=10,\epsilon=7$. 
\begin{align*}
    \bx &= \balpha \overbrace{0000000}^{0^\epsilon}\makebox[0pt]{\text{$|$}}\overbrace{000\textcolor{red}{\zero\zero\zero\zero\zero\zero\zero}}^{0^k}\overbrace{1111111111}^{1^k}\bbeta\\
    \by &=\balpha|\underbrace{0000000000}_{0^k}\underbrace{\textcolor{red}{\one\one\one\one\one\one\one}111}_{1^k}\underbrace{1111111}_{1^\epsilon}
\bbeta
\end{align*}
In this case, we can see 
        $N_{1^k}^{L}(\bx)=N_{1^k}^{L}(\by)$, and because of the $1^{k+\epsilon}$ in $\by$ compared to the $1^k$ in $\bx$, we have $N_{1^k}^{R}(\bx)=N_{1^k}^{R}(\by)-\epsilon.$ So $N_{1^k}(\bx)\neq N_{1^k}(\by)$, and thus $\Delta N_{1^k}\neq 0$.

Second, suppose $\bx=\balpha|0^{k}1^k0^{k-\epsilon}1^k\bbeta$ and $\by=\balpha0^{k}1^{k-\epsilon}|0^{k}1^k\bbeta$. Here $0<\epsilon< k$ (when $\epsilon=k$, this edit is also described by the previous case). The burst of substitutions occurs at positions $[\abs{\balpha}+2k-\epsilon+1,\abs{\balpha}+2k]$. An example with $k=10,\epsilon=7$ is shown below.
\begin{align*}
    \bx &= \balpha \makebox[0pt]{\text{$|$}}\overbrace{0000000000}^{0^k}\overbrace{111\textcolor{red}{\one\one\one\one\one\one\one}}^{1^k}\overbrace{000}^{0^{k-\epsilon}}\overbrace{1111111111}^{1^k}\bbeta\\
    \by &=\balpha\underbrace{0000000000111}_{0^k1^{k-\epsilon}}\makebox[0pt]{\text{$|$}}\underbrace{\textcolor{red}{\zero\zero\zero\zero\zero\zero\zero}000}_{0^k}\underbrace{1111111111}_{1^k}
\bbeta
\end{align*}
        In this case, we can see 
        $N_{1^k}^{L}(\bx)=N_{1^k}^{L}(\by)$, and because there are two $1^k$'s between `$|$' and $\beta$ in $\bx$ compared to only one $1^k$ between `$|$' and $\beta$ in $\by$, we have 
 $N_{1^k}^{R}(\bx)=N_{1^k}^{R}(\by)+1.$ So $N_{1^k}(\bx)\neq N_{1^k}(\by)$, and thus $\Delta N_{1^k}\neq 0$.
\end{example}
It can be shown that the two cases in the above example are the only two key cases in which a burst of substitutions shifts a pattern (see full version).
We can now formally construct $f$.
For a string $\bx$, let both $l(\bx)$ and $L(\bx)$ represent the length of $\bx$, let
$n_{1^k}(\bx)$ be the number of $1^k$ inside $\bx$, and let $d(\bx)$ be as in the second attempt. Let 
$B=3 (2^k)$, and $C=40k(2^k)$, and define 
\begin{align}\label{eq:fconstruction}
   f(\bx):= d(\bx)+Bn_{1^k}(\bx)+Cl(\bx). 
\end{align}
If $\bx\sseneq \by$, then $\bsf(\bx)\Rightarrow\bsf(\by)$ satisfies the second condition of locatability. It is not hard to prove that $\bsf(\bx)\Rightarrow\bsf(\by)$ also satisfies the first condition of locatability (see full version). In summary, we have the following Lemma.

\begin{lemma}\label{lem:locatable}
    If $\bx\sseneq \by$, then $\bsf(\by)$ is obtained from $\bsf(\bx)$ by a locatable $3$-substring edit.
\end{lemma}

\begin{lemma}\label{lem:8deltaCstring}
    If $\bx$ is $(\cP,\delta)$-dense and $\bx\rightarrow\by $, then $\bsf(\bx)$ and $\bsf(\by)$ are $8\delta C$-bounded strings.
\end{lemma}

\paragraph{Correcting an edit with known approximate location} 
We start by adapting a result on systematic codes correcting multiple edits \cite{sima2020optimal_Systematic} to a hash correcting a $k$-substring edit, since  a $k$-substring edit can also be viewed as $k$ edits. The adapted result states that a substring edit in a string of length $n$ can be corrected using $4k\log n+o(\log n)$ bits of redundancy. Although our goal is to achieve about $\log n$ redundancy, this result will be useful in our construction.
\begin{lemma}[c.f.~\cite{sima2020optimal_Systematic}]\label{fact:syndrome}
    There exists a hash $\phi_n:\{0,1\}^n\rightarrow \{0,1\}^{4k\log n+o(\log n)}$ that can correct a single $k$-substring edit. Namely, if $\br\in \{0,1\}^n$ and $\br\rightarrow\bs$, then $\br$ can be recovered from $\bs$ and $\phi_n(\br)$.
\end{lemma}

 Let $\psi_A(\bx):\{0,1\}^n\rightarrow [2^{4k\log A+o(\log A)}]$ be defined as
\begin{multline*}
    \psi_A(\bx) = \left(\left(\sum_{i \text{ odd}}\phi_A(\bx_i)\right)\bmod 2^{4k\log A+o(\log A)},\right.\\\left.\left(\sum_{i \text{ even}}\phi_A(\bx_i)\right)\bmod 2^{4k\log A+o(\log A)}\right),
\end{multline*}
 where the $\bx_i$ result from partitioning $\bx$ into parts of length $A$ 
 (note that in the definition of the the map $\psi_A(\bx)$ we are not partitioning $\bx$ according to $\bx^{\cP}$)
 . 

\Cref{cla:oddeven} is an adaptation of a technique used in \cite{bitar2021optimal} to our substring edit scenario, which will  allow us to correct a substring edit whose position is approximately known.
\begin{lemma}[c.f.~\cite{bitar2021optimal}]\label{cla:oddeven}
    If $\bx\xrightarrow[]{} \by$, then one can recover $\bx$ given $\by$ and $\psi_A(\bx)$ and an interval of length $A$ that contains the edit position.
\end{lemma}

\subsection{Substring edit-correcting code and its redundancy}
We now put the components previously discussed together to construct the error-correcting code. Let $\cP$ be as in Section~\ref{subsec:components}, and let $f$ and $C$ be the one in equation~\eqref{eq:fconstruction}. Let 
\begin{align*}    
h(\bx):=&\left(\VT(\bsf(\bx))\bmod 2000Cn, \vphantom{\sum\nolimits_{i=1}^{n_{\cP}(\bx)}}\right.\\ &\quad\left.\left(\sum\nolimits_{i=1}^{n_{\cP}(\bx)}f_i(\bx)\right) \bmod 10Ck, n_{\cP}(\bx)\bmod 5 \right).
\end{align*}

For $c_1,c_2\in \mathbb{N}$, let 
\begin{align*}
        \cC=\{\bx \text{ is } (\cP, \delta) \text{-dense: }   h(\bx)=c_1, 
         \psi_{O(\delta^2)}(\bx)=c_2 \}.
    \end{align*}
    The next lemma will be used to show that the redundancy arising from restricting $\bx$ to be $(\cP, \delta)\text{-dense}$ is negligible.
\begin{lemma}\label{cla:manydense}
Let $\delta=k2^{2k+3}\log n$. For a uniformly randomly chosen $\bx\in \{0,1\}^n$, $$\Pr[\bx \text{ is not } (\cP,\delta)\text{-dense}]<1/n.$$
\end{lemma}

\begin{theorem}\label{thm:1substringedit}
    $\cC$ corrects one $k$-substring edit. Moreover, for $\delta=k2^{2k+3}\log n$,  it has redundancy $\log n+O(\log\log n)$. 
\end{theorem}

\begin{IEEEproof}[Proof sketch]
    We first prove the redundancy. Note that $h(\bx)$ needs $\log n+O(1)$ redundancy and $\psi_{O(\delta^2)}(\bx)$ needs $O(\log \log n)$ redundancy. By~\Cref{cla:manydense}, restricting $\bx$ to be $(\cP, \delta)$-dense adds $O(1)$ redundancy. So $\cC$ has redundancy $\log n+O(\log \log n)$.

To prove the error-correction ability of the code, we first show that if $\bx\to\by$, given $\by$ and $h(\bx)$, we can determine whether $\bx=\by$. By Lemma~\ref{lem:locatable} and the definition of locatable substring edit, $\bx=\by$ if and only if $\sum\nolimits_{i=1}^{n_{\cP}(\bx)}f_i(\bx)\equiv \sum\nolimits_{i=1}^{n_{\cP}(\by)}f_i(\by)\mod 10Ck$ and $n_{\cP}(\bx)\equiv n_{\cP}(\by)\mod 5$. These two equations can be checked given $\by$ and $h(\bx)$ to determine whether $\bx=\by$. If $\bx=\by$, then we just recover $\bx$ by outputting $\by$.

We will now assume $\bx\sseneq \by$ and show that we can determine the edit position to within an interval of length $O(\delta^2)$. By~\Cref{lem:locatable,lem:8deltaCstring}, we can treat $\bsf(\bx)$ and $\bsf(\by)$ as the $\bz$ and $\bw$ in Lemma~\ref{lem:locationA}, and determine the position of the locatable edit to within an interval of length $O(\delta)$. Since $\bx$ is $(\cP,\delta)$-dense, we can further determine the position of the original edit to within an interval of length $O(\delta^2)$. In the final step, by~\Cref{cla:oddeven}, we can recover $\bx$ by $\psi_{O(\delta^2)}(\bx)$.
\end{IEEEproof}

\section*{Acknowledgment}
This work was supported in part by NSF grants CIF-2312871, CIF-2312873, CIF-2144974, and CCF-2212437.

\bibliographystyle{IEEEtran}
\bibliography{bibliofile}

\end{document}